\documentclass[12pt]{llncs}
\usepackage{amsfonts}
\usepackage{psfrag}

\usepackage[utf8]{inputenc} 
\usepackage{amsmath}
\usepackage{amssymb}
\usepackage{tikz}

\title{\hspace*{-27.55pt}\mbox{The~Approximation~Ratio~of~the~Greedy~Algorithm} for the Metric Traveling Salesman Problem}
\author{Judith Brecklinghaus and Stefan Hougardy}
\institute{Research Institute for Discrete Mathematics,
           University of Bonn\\
           Lenn\'estr.~2, 53113 Bonn, Germany\\[5mm]
           \today}
\begin{document}
\maketitle 

\begin{abstract}
We prove that the approximation ratio of the greedy algorithm for the metric
Traveling Salesman Problem is $\Theta(\log n)$. % for instances with $n$-cities.
Moreover, we prove that the same result also holds for graphic, euclidean, and
rectilinear instances of the Traveling Salesman Problem.   
Finally we show that the approximation ratio of the Clarke-Wright savings heuristic
for the metric Traveling Salesman Problem is $\Theta(\log n)$.
\end{abstract}

{\small\textbf{keywords:} traveling salesman problem; greedy algorithm, Clarke-Wright savings heuristic, approximation algorithm}

\section{Introduction}

Let $G=(V,E)$ be an undirected complete graph and $c:E(G)\to \mathbb{R}_{\ge 0}$ be a length function on the edges. 
A \emph{tour} is a cycle in $G$ that visits each vertex exactly once. The \emph{length} of a tour $T$
is the sum of the lengths of the edges in $T$. The Traveling Salesman Problem (TSP) is to find a tour of minimum length.
In this paper we study the \emph{metric} Traveling Salesman Problem which is the special version of the 
Traveling Salesman Problem where the function $c$ is metric, i.e., we have $c(\{x,y\}) + c(\{y,z\}) \ge c(\{x,z\})$
for any three vertices $x,y,z\in V(G)$.  
The metric Traveling Salesman Problem is known to be NP-hard~\cite{Kar1972} 
and therefore much effort has been spent to design polynomial time algorithms that find good tours.

An algorithm $A$ for the traveling salesman problem has \emph{approximation ratio}
$c$ if for every TSP instance it finds a tour that is at most 
$c$ times longer than a shortest tour. 
The \emph{greedy algorithm} is one of the simplest algorithms to find a TSP tour. It starts
with an empty edge set and adds in each iteration the cheapest edge (with respect to the length function $c$),
such that the resulting graph is a subgraph of a tour. 
Because of its simplicity 
and because the greedy algorithm achieves quite good results on real world instances (see e.g.~\cite[page 98]{Rei1994}) 
it is often used in practice. 

In Section~\ref{Sec:Greedy} we will prove that the approximation ratio of the greedy algorithm for
metric TSP is $\Theta(\log n)$ for instances with $n$ vertices. Our result closes the long-standing
gap between the so far best known lower bound of $\Omega(\log n/ \log\log n)$ and the upper bound of $O(\log n)$.
Both these bounds were proved in 1979 by Frieze~\cite{Fri1979}.

Our result also holds for the \emph{euclidean}, the \emph{rectilinear}, and the \emph{graphic} TSP.
These are well studied special cases of the metric TSP. In the euclidean and the rectilinear TSP
the cities are points in the plane and the distance between two cities is defined as the euclidean respectively
rectilinear distance. A graphic TSP is obtained from an unweighted, undirected, and 
connected graph $G$ which has as vertices all the cities. 
The distance between two cities is then defined as the length of a shortest path
in $G$ that connects the two cities.

Another well established approximation algorithm for the traveling salesman problem 
that achieves good results in practice (see e.g.~\cite[page 98]{Rei1994}) is the 
Clarke-Wright savings heuristic~\cite{CW1964}. This heuristic selects one city $x$ and connects it by two parallel edges 
to each other city. This way one obtains a Eulerian tour that is transformed into a TSP tour as follows.
For each pair of cities $a,b$ let the \emph{savings} be the amount by which the edge $\{a,b\}$ is shorter than 
the sum of the lengths of the two edges $\{a,x\}$ and $\{b,x\}$. The replacement of the edges 
$\{a,x\}$ and $\{b,x\}$ by the edge $\{a,b\}$ is called a \emph{short-cut}.
Now go through all city pairs in 
nonincreasing order of savings and short-cut the pair if this does not close a cycle on the cities different
from $x$ and if both cities in the pair do not get adjacent to more than two other cities.
After going through all city pairs the result will be a TSP tour.

In Section~\ref{Sec:CW} we will prove that the approximation ratio of the Clarke-Wright savings heuristic
for metric TSP is $\Theta(\log n)$. Again, we close a long-standing gap between the 
so far best known lower bound of $\Omega(\log n/ \log\log n)$ and the upper bound of $O(\log n)$.
Both these bounds were proved in 1984 by Ong and Moore~\cite{OM1984}.

\section{The Approximation Ratio of the Greedy Algorithm}\label{Sec:Greedy}

In this section we describe the construction of a family of metric TSP instances $G_k$ on which
the greedy algorithm can find a TSP tour that is by a factor of $\Omega(k)$ longer than an optimum tour.
Our construction is similar to the approach in~\cite{HW2014+}.

We denote by $V_k$ the set of cities in $G_k$. As $V_k$ we take
the points of a $2\times (\frac12 (3^{k+2}-1))$ subgrid of $\mathbb{Z}^2$.
Thus we have 
\begin{equation} \label{eq:Vk}
|V_k|~=~ 3^{k+2}-1 
\end{equation}

Let $G_k$ be any TSP-instance defined on the cities $V_k$ that satisfies the following conditions:
\begin{itemize}
\item[(i)] if two cities have the same x-coordinate or the same y-coordinate
           their distance is the euclidean distance between the two cities.
\item[(ii)] if two cities have different x-coordinate and  different y-coordinate
           then their distance is at least as large as the absolute difference between
           their x-coordinates.       
\end{itemize}

Note that if we choose as $G_k$ the euclidean or the rectilinear TSP instance on $V_k$ 
then conditions (i) and (ii) are satisfied.
We can define a graph on $V_k$ by adding an edge between each pair of cities
at distance 1. The graphic TSP that is induced by this graph is exactly the rectilinear TSP on $V_k$.
The graph for $G_0$ is shown in Figure~\ref{fig:G0}. 
We label the vertex at position $\frac12(3^{k+1}+1)$ in the top row of $V_k$ as $s_k$ 
and the right most vertex in the bottom row of $V_k$ as $r_k$.

\begin{figure}[ht]
\centering
\begin{tikzpicture}[scale=1]
\draw[step = 1cm] (0,0) grid (3,1); 
\draw[blue,very thick] (0,0) -- (0,1);
\draw[blue,very thick] (0,1) -- (1,1);
\draw[blue,very thick] (0,0) -- (1,0);
\draw[blue,very thick] (1,0) -- (2,0);
\draw[blue,very thick] (2,0) -- (2,1);
\draw[blue,very thick] (2,1) -- (3,1);
\draw[blue,very thick] (3,1) -- (3,0);
\foreach \x in {0,1,2,3}
	\foreach \y in {0,1}
		\fill (\x, \y) circle(1mm);

\draw (3,0) node[anchor = west] {$r_0$};
\draw (1,1) node[anchor = south] {$s_0$};
\end{tikzpicture}
\caption{The graph defining the graphic TSP $G_0$ together with a partial greedy tour (bold edges) that 
connects $s_0$ with $r_0$.}
\label{fig:G0}
\end{figure}
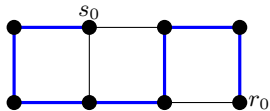

A TSP tour is called \emph{greedy tour} if this tour can be the output of the greedy algorithm.
Let $V$ be a subset of the cities $V_k$ of the instance $G_k$ and let $T$ be a greedy tour in $G_k$. 
A \emph{partial greedy tour} on $V$ is a path $P$ containing exactly the vertices of $V$ such that there exists
a greedy tour in $G_k$ that contains $P$.
The bold edges in Figure~\ref{fig:G0} form a partial greedy tour on $V_0$.
The next lemma is similar to Lemma~1 in~\cite{HW2004} and Lemma~1 in~\cite{HW2014+}.

\begin{lemma}
\label{lem:Gk}
Let the cities of $G_k$ be embedded into $G_m$ with $m \ge k$. Then there exists a partial greedy tour $P$
in $G_m$ that 
\begin{itemize}
\item[\rm(a)] contains exactly the cities in $G_k$,
\item[\rm(b)] connects $s_k$ and $r_k$, 
\item[\rm(c)] has length exactly $(2k+8)\cdot 3^k-1$, and
\item[\rm(d)] the edges in $P$ have length $3^i$ with $0\le i\le k$.
\end{itemize}
\end{lemma}

\begin{proof}
We use induction on $k$ to prove the lemma. For $k=0$ a
partial greedy tour of length $8\cdot 3^0-1=7$ that satisfies (a)--(d)
is shown in Figure~\ref{fig:G0}.
Now assume we already have defined a partial greedy tour for $G_k$. 
Then we define a partial greedy tour for $G_{k+1}$ recursively as follows.
By (\ref{eq:Vk}) we have
 $|V_{k+1}| = 3^{k+3}-1 = 3\cdot(3^{k+2}-1) + 2$. Therefore, 
we can think of $G_{k+1}$ to be the disjoint union of three copies $G_k'$, $G_k''$, and $G_k'''$ of $G_k$
such that $G_k'$ and $G_k''$ are separated by a $2\times 1$ grid. Moreover, we assume
that $G_k''$ is embedded after mirroring at the y-axis. This is shown in Figure~\ref{fig:Gk}.

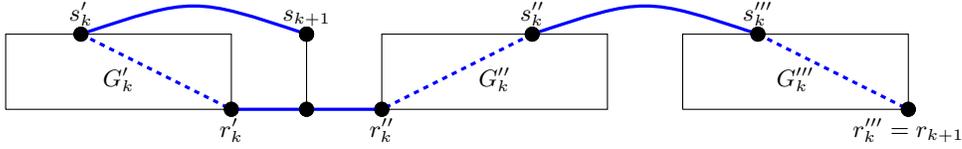
\begin{figure}[ht]
\centering
\begin{tikzpicture}
%box for G_k'
\draw (0,0) -- (0,1) -- (3,1) -- (3,0) -- (0,0) ;
\draw (1.5,0.4) node {$G'_k$};
\fill (1,1) circle(0.1); \draw (1,1) node[anchor = south] {$s_k'$};
\fill (3,0) circle(0.1); \draw (3,0) node[anchor = north] {$r_k'$};

%2x1 grid
\draw (4,0) -- (4,1) ;
\fill (4,0) circle(0.1);
\fill (4,1) circle(0.1);

%box for G_k''
\draw (5,0) -- (5,1) -- (8,1) -- (8,0) -- (5,0) ;
\draw (6.5,0.4) node {$G''_k$};
\fill (7,1) circle(0.1); \draw (7,1) node[anchor = south] {$s_k''$};
\fill (5,0) circle(0.1); \draw (5,0) node[anchor = north] {$r_k''$};

%box for G_k'''
\draw (9,0) -- (9,1) -- (12,1) -- (12,0) -- (9,0) ;
\draw (10.5,0.4) node {$G'''_k$};
\fill (10,1) circle(0.1); \draw (10,1) node[anchor = south] {$s_k'''$};
\fill (12,0) circle(0.1); \draw (12,0) node[anchor = north] {$r_k''' = r_{k+1}$};

\draw[dash pattern = on 2pt off 2pt,blue, very thick] (1,1) -- (3,0);
\draw[dash pattern = on 2pt off 2pt,blue, very thick] (7,1) -- (5,0);
\draw[dash pattern = on 2pt off 2pt,blue, very thick] (10,1) -- (12,0);
\draw[blue, very thick] (3,0) -- (4,0) -- (5,0);

\draw[blue, very thick] (1,1) .. controls (2.5,1.5) .. (4,1);
\draw[blue, very thick] (7,1) .. controls (8.5,1.5) .. (10,1);

\fill (4,1) circle(0.1); \draw (4,1) node[anchor = south] {$s_{k+1}$};

% redraw points
\fill (1,1) circle(0.1); 
\fill (3,0) circle(0.1); 
\fill (4,0) circle(0.1);
\fill (4,1) circle(0.1);
\fill (7,1) circle(0.1); 
\fill (5,0) circle(0.1); 
\fill (10,1) circle(0.1);
\fill (12,0) circle(0.1);
\fill (4,1) circle(0.1); 

\end{tikzpicture}
\caption{The recursive construction of a partial greedy tour for the instance $G_{k+1}$. The dashed lines 
indicate partial greedy tours in $G_k'$, $G_k''$, and $G_k'''$. }
\label{fig:Gk} 
\end{figure}

By definition $s_{k+1}$ is the vertex at position $\frac12(3^{k+2}+1)$ in the top row of $V_{k+1}$.
As we have  $\frac12(3^{k+2}+1) = \frac12(3^{k+2}-1)+1 =  \frac12|V_k|+1 $ this is the first vertex in the top row
behind $G_k'$. 

We now claim that a partial greedy tour for $G_{k+1}$ may look as follows (see the bold edges in Figure~\ref{fig:Gk}):
it contains the three recursively defined partial greedy tours in $G_k'$, $G_k''$, and $G_k'''$ plus four additional edges.
These four additional edges are the two edges of length one that leave $r_k'$ and $r_k''$ to the right respectively to
the left plus the two edges $\{s_k',s_{k+1}\}$ and $\{s_k'',s_k'''\}$. 
As  $s_k'$ is at position $\frac12(3^{k+1}+1)$ in the top row and $s_{k+1}$ is at position $\frac12(3^{k+2}+1)$ in the top row,
we have that the edge $\{s_k',s_{k+1}\}$ has length 
$$ \frac12(3^{k+2}+1) - \frac12(3^{k+1}+1) = 3^{k+1}.$$
The edge $\{s_k'',s_k'''\}$ has length
$$\frac12(3^{k+1}+1) + \frac12(3^{k+1}+1) -1 = 3^{k+1}.$$ 

Thus, we have constructed a partial tour that satisfies (a), (b), and (d) and its total length is
$$ 3\cdot ((2k+8)\cdot 3^k-1) + 2 + 2\cdot 3^{k+1} = (2(k+1)+8)\cdot 3^{k+1}-1 $$
and therefore also (c) holds. We still have to prove that this partial tour is contained in some greedy tour.
The two edges of length one that leave $r_k'$ and $r_k''$ to the right respectively to
the left may be chosen by the greedy algorithm at the very beginning. Afterwards, by induction
the partial greedy tours within $G_k'$, $G_k''$, and $G_k'''$ may be chosen. By induction, the
longest edge that has to be considered by the greedy algorithm up to this step has length $3^k$. 
We now claim that all edges that are incident to the vertices $s_k', s_k'', s_k'''$, and $s_{k+1}$  and that
may be added by the greedy algorithm have length at least $3^{k+1}$. For the vertices 
$s_k'', s_k'''$, and $s_{k+1}$ this is immediately clear, as any vertex not belonging to $G_{k+1}$ has distance at least 
$3^{k+1}$ to these vertices. We still have to rule out that $s_k'$ has a neighbor to the left of $G_{k+1}$ within a distance
smaller than $3^{k+1}$. But if there exists a neighbor to the left of $G_{k+1}$ then by the recursive construction
$G_{k+1}$ is contained in an embedding of $G_{k+2}$ and this implies that to the left of $G_k'$ there is
a mirrored copy of $G_k$. Thus the next neighbor to the left of $s_k'$ has distance at least $3^{k+1}$.
\end{proof}

\begin{theorem}
\label{thm:main}
On graphic, euclidean, and rectilinear TSP instances with $n$ cities
the approximation ratio of the greedy algorithm is $\Theta(\log n)$.
\end{theorem}

\begin{proof}
The upper bound is proven in~\cite{Fri1979}. For the lower bound we apply Lemma~\ref{lem:Gk}.
The instance $G_k$ defined above has $n:=3^{k+2}-1$ cities and an optimum TSP tour 
in $G_k$ has length $n$. As shown in Lemma~\ref{lem:Gk} there exists a partial greedy tour in $G_k$ 
of length at least $(2k+8)\cdot 3^k-1$. Thus, for $k\ge 1$ the approximation 
ratio of the greedy algorithm is greater than  
$$\frac {(2k+8)\cdot 3^k-1}{3^{k+2}-1} ~\ge~ \frac {(2k+8)\cdot 3^k}{3^{k+2}} ~=~\frac{2k+8}{9}
~\ge~\frac29 \cdot \log_3 (n + 1).$$
\end{proof}

Conditions (i) and (ii) are satisfied whenever
the distances in $G_k$ are defined by an $L^p$-norm. Thus Theorem~\ref{thm:main} not only holds
for the $L^2$- and the $L^1$-norm but for all $L^p$-norms.

\subsection{The 1-2-TSP}

For completeness we also provide the approximation ratio of the greedy algorithm for the 1-2-TSP.

\begin{theorem}
\label{thm:1-2-tsp}
The approximation ratio of the greedy algorithm for the 1-2-TSP is $\frac32 - \frac{1}{2n}$
on instances with $n$ cities.
\end{theorem}

\begin{proof}
We first present an example showing that the approximation ratio is at least $\frac32 - \frac{1}{2n}$.
For an odd number $n$ with $n > 4$ consider the cycle on vertices $\{1,\ldots, n\}$. Add to this cycle
all edges $\{i,j\}$ with $|i-j|=2$ and $i$ and $j$ odd. All these edges get length one while all other edges get length~2.
Now the greedy algorithm may choose $(n+1)/2$ edges of length one by taking the edges 
$\{2,3\}$ and $\{1,n\}$ in addition to the edges $\{3,5\}, \{5,7\}, \ldots, \{n-2, n\}$.
No other edge of length one can be chosen by the greedy algorithm. Thus the total length of the tour returned
by the greedy algorithm is $(n+1)/2 + 2(n-1)/2 = (3n-1)/2$. As an optimum tour obviously has length $n$,
the approximation ratio of the greedy algorithm is at least $\frac32 - \frac{1}{2n}$. 

Now we prove that the approximation ratio of the greedy algorithm is never worse than $\frac32 - \frac{1}{2n}$.
Let $T$ be a tour found by the greedy algorithm and denote by $G_1$
the subgraph of $T$ that contains all edges of length one. Let $m$ be the number of edges in $G_1$.
We may assume that $T$ is not an optimum tour and therefore $G_1$ will consist of several connected components,
each of which must be a path. 
Let $e$ be an edge of length one contained in an optimum TSP tour. 
As $e$ was not added by the greedy algorithm, it must either be incident to a vertex of degree two in $G_1$ or
it connects two vertices of degree one in the same connected component of $G_1$.
For each connected component $C$ of $G_1$ of size at least two we can have at most
$2(|C|-2)+1 = 2(|C|-1)-1$ edges of the optimum tour that have length one and are either incident to a vertex of degree two in $C$
or connect two vertices of degree one in $C$. 
By summing over all connected components of size at least two in $G_1$ we obtain $k \le 2 m -1$ 
where $k$ is the number of length one edges in an optimum TSP tour.
As the greedy tour has length $2n-m \le 2n - (k+1)/2$ and the optimum tour has length $2n-k$
we conclude that the approximation ratio of the greedy algorithm is at most 
$$\frac{2n-(k+1)/2}{2n-k} \le  \frac{2n-(n+1)/2}{2n-n} = \frac32 - \frac{1}{2n}.$$
\end{proof}

\section{The Approximation Ratio of the Clarke-Wright Savings Heuristic}
\label{Sec:CW}

\begin{theorem}
\label{thm:CW}
For metric TSP instances with $n$ cities
the approximation ratio of the Clarke-Wright savings heuristic is $\Theta(\log n)$.
\end{theorem}

\begin{proof}
The upper bound is proven in~\cite{OM1984}. For the lower bound extend the construction used
to prove Lemma~\ref{lem:Gk}. Let $G_k$ be the graph as defined above and 
let $x$ be a new vertex that is connected to all vertices in $G_k$. 
Take the graphic instance induced by $G_k$ and define the length of all edges incident to $x$
as $\frac12\cdot 3^{k+2}$. Note that this value is larger than any distance
within the graph $G_k$. 

Now start the Clarke-Wright savings heuristic on this instance with city $x$. 
As all edges incident to $x$ have the same length, the short-cutting of a pair results in a reduction of the edge length that is
only dependent on the length of the edge between the two cities in the pair. Therefore, the Clarke-Wright savings heuristic
may choose the pairs to short-cut exactly in the same order as the greedy algorithm applied to $G_k$. 
Thus, we just have taken into account that the greedy tour and an optimum tour contain two edges of total length 
$3^{k+2}$ incident with city $x$. Similarly as in the proof of Theorem~\ref{thm:main} 
we get that for $k\ge 1$ the approximation 
ratio of the Clarke-Wright savings heuristic is greater than  
$$\frac {(2k+8)\cdot 3^k-1 + 3^{k+2}}{3^{k+2}-1 + 3^{k+2}} ~\ge~ \frac {(2k+17)\cdot 3^k}{2\cdot 3^{k+2}} ~=~\frac{2k+17}{18}
~\ge~\frac19 \cdot \log_3 n.$$ 
\end{proof}

\bibliographystyle{plain}
\bibliography{Literatur}

\end{document}